\documentclass[%
a4paper,
UKenglish,
cleveref,
autoref,
thm-restate,
]{lipics-v2021}

\hideLIPIcs  
\nolinenumbers

\title{A Simple Algorithm for Worst Case Optimal Join and Sampling} 

\titlerunning{A Simple Algorithm for WCO Join and Sampling}

\author{Florent Capelli}{%
Université d'Artois, CNRS, UMR 8188 - CRIL, F-62300 Lens, France
}{%
  florent.capelli@univ-artois.fr%
}{%
  https://orcid.org/0000-0002-2842-8223%
}{}

\author{Oliver Irwin}{%
  Université de Lille, CNRS, Inria, UMR 9189 - CRIStAL, F-59000 Lille, France%
}{%
  oliver.irwin@univ-lille.fr%
}{%
  https://orcid.org/0000-0002-8986-1506%
}{}

\author{Sylvain Salvati}{%
  Université de Lille, CNRS, Inria, UMR 9189 - CRIStAL, F-59000 Lille, France%
}{%
  sylvain.salvati@univ-lille.fr%
}{%
}{}

\authorrunning{F. Capelli, O. Irwin and S. Salvati} 

\Copyright{Florent Capelli, Oliver Irwin and Sylvain Salvati} 



\category{} 

\funding{
  This work was supported by project ANR KCODA, ANR-20-CE48-0004.%
}

\acknowledgements{} 

\EventEditors{John Q. Open and Joan R. Access}
\EventNoEds{2}
\EventLongTitle{42nd Conference on Very Important Topics (CVIT 2016)}
\EventShortTitle{CVIT 2016}
\EventAcronym{CVIT}
\EventYear{2016}
\EventDate{December 24--27, 2016}
\EventLocation{Little Whinging, United Kingdom}
\EventLogo{}
\SeriesVolume{42}
\ArticleNo{23}

\usepackage{gensymb}
\usepackage{mathtools}
\usepackage{hyperref}
\usepackage{cleveref}
\usepackage{algorithm}
\usepackage{algpseudocode}
\usepackage{todonotes}
\usepackage{tcolorbox}
\usepackage{array}
\usepackage{xcolor}
\usepackage{caption}
\usepackage{subcaption}
\usepackage[inline]{enumitem}

\usepackage[cleveref, electronic]{knowledge}

\knowledgeconfigure{notion, quotation, diagnose line=true}

\NewDocumentCommand{\C}{}{\ensuremath{\mathcal{C}}}
\NewDocumentCommand{\CCard}{O{\tup[N]} O{H}}{\ensuremath{\C_{#2}(\leq #1)}}
\NewDocumentCommand{\CDeg}{O{DC} O{H}}{\ensuremath{\C_{#2}({#1})}} 

\NewDocumentCommand{\N}{}{\ensuremath{\mathbb{N}}}

\NewDocumentCommand{\tup}{O{x}}{\ensuremath{\mathbf{#1}}}
\NewDocumentCommand{\emptytuple}{}{\ensuremath{\langle \rangle}}
\NewDocumentCommand{\unittuple}{O{x} O{d}}{\ensuremath{\langle #1 \gets #2 \rangle}}

\NewDocumentCommand{\ans}{O{Q}}{\ensuremath{\mathsf{ans}(#1)}}
\NewDocumentCommand{\bin}{O{R} O{b}}{\ensuremath{\widetilde{#1}^{#2}}}
\NewDocumentCommand{\debin}{O{\cdot}}{\ensuremath{\overline{#1}}}
\NewDocumentCommand{\dsize}{O{Q}}{\ensuremath{\|#1\|}}
\NewDocumentCommand{\poly}{O{Q}}{\ensuremath{\mathsf{poly}}}
\NewDocumentCommand{\polylog}{O{Q}}{\ensuremath{\mathsf{polylog}}}
\NewDocumentCommand{\wc}{}{\ensuremath{\mathsf{wc}}}
\NewDocumentCommand{\up}{}{\ensuremath{\mathsf{upb}}}
\NewDocumentCommand{\qup}{}{\ensuremath{\mathsf{q\_upb}}}
\NewDocumentCommand{\agmup}{}{\ensuremath{\mathsf{agm\_upb}}}
\NewDocumentCommand{\pmup}{}{\ensuremath{\mathsf{pm\_upb}}}
\NewDocumentCommand{\eIf}{m m}{\State {\bf if} #1 {\bf then} #2}
\NewDocumentCommand{\eFor}{m m}{\State {\bf for} #1 {\bf do} #2}
\RenewDocumentCommand{\leq}{}{\leqslant}
\RenewDocumentCommand{\geq}{}{\geqslant}
\RenewDocumentCommand{\log}{}{\mathsf{log}}
\NewDocumentCommand{\WCOJ}{}{worst-case optimal join}
\NewDocumentCommand{\WCJ}{}{\ensuremath{\mathsf{WCJ}}}
\NewDocumentCommand{\size}{ m }{\ensuremath{\vert#1\vert}}
\NewDocumentCommand{\onel}{ O{T} }{\ensuremath{\mathsf{leaves_1}(#1)}}
\NewDocumentCommand{\Size}{ m }{\ensuremath{\big\vert#1\big\vert}}
\NewDocumentCommand{\linp}{ O{\mathcal{H}}}{\ensuremath{\mathsf{lin}(#1)}}

\NewDocumentCommand{\children}{}{\ensuremath{\mathsf{children}}}

\NewDocumentCommand{\UP}{}{\ensuremath{\mathsf{UP}}}
\NewDocumentCommand{\OUT}{}{\ensuremath{\mathsf{OUT}}}

\NewDocumentCommand{\bigo}{}{\ensuremath{\mathcal{O}}}
\NewDocumentCommand{\bigot}{}{\ensuremath{\mathcal{\tilde O}}}

\RenewDocumentCommand{\max}{}{\ensuremath{\mathsf{max}}}

\NewDocumentCommand{\flo}{ O{} m }{%
  \todo[linecolor=blue,
        backgroundcolor=blue!25,
        bordercolor=blue,
        tickmarkheight=0.15cm,
        #1]{
        \textbf{Florent : }#2}
}

\NewDocumentCommand{\oli}{ O{} m }{%
  \todo[linecolor=green,
  backgroundcolor=green!25,
  bordercolor=green,
  tickmarkheight=0.15cm,
  #1]{
    \textbf{o : }#2}
}

\NewDocumentCommand{\sylvain}{ O{} m }{%
  \todo[linecolor=orange,
  backgroundcolor=orange!25,
  bordercolor=orange,
  tickmarkheight=0.15cm,
  #1]{
    \textbf{Sylvain : }#2}
}

\input{notions.kl}

\begin{document}

\IfKnowledgeCompositionModeTF{
\knowledgestyle{kl}{color=black}
\knowledgestyle{notion}{color=black}
\knowledgestyle{intro notion}{color=black, emphasize}
}{}

\maketitle

\begin{abstract}
We present an elementary branch and bound algorithm with a simple analysis of why it achieves worstcase optimality for join queries on classes of databases defined respectively by cardinality or acyclic degree constraints. We then show that if one is given a reasonable way for recursively estimating upper bounds on the number of answers of the join queries, our algorithm can be turned into algorithm for uniformly sampling answers with expected running time $\bigot(\UP/\OUT)$ where $\UP$ is the upper bound, $\OUT$ is the actual number of answers and $\bigot(\cdot)$ ignores polylogarithmic factors. Our approach recovers recent results on worstcase optimal join algorithm and sampling in a modular, clean and elementary way.
\end{abstract}

\section{Introduction}
\label{sec:introduction}


\emph{Join queries} are expressions of the form \(Q \coloneq
R_1(\mathbf{x_1}),\dots,R_m(\mathbf{x_m})\), where every \(R_i\) is a
relation symbol and the \(\mathbf{x_i}\) is a tuple of variables over a set
\(X\). %
Evaluating join queries is a central task when answering database
queries. %
Since the combined complexity of deciding whether a given join query has at
least one answer on a given database is NP-complete~\cite{chandra77}, it is
unlikely that one can list all its answers in time linear in the number of
answers. %
An interesting line of research has been the design of so called worst case
optimal join (\emph{WCOJ}) algorithms. %
In this setting, for a given query \(Q\), we consider the worst possible
database among a class of instances, that is, the one where the number of
answers of \(Q\) is maximal. %
Now, even if we cannot find the answers of \(Q\) in time linear in the
number of answers of \(Q\), we can still aim at finding every answer in
time linear in the number of answers of the worst possible database in the
class. %
Such an algorithm will be said to be a WCOJ algorithm.

Consider for example the \emph{triangle query}, a query we will use multiple
times in this paper for illustration:

\begin{equation}
  \label{eq:q_delta}
  Q_\Delta \coloneq R(x_1, x_2), S(x_2, x_3), T(x_1, x_3)
\end{equation}

We assume that \(R, S, T\) are relations of size respectively \(N_R, N_S\)
and \(N_T\). %
It is not hard to see that \(Q_\Delta\) will have never more than \(N_R \times N_S \times
N_T\) answers. %
Even better, one can notice that the variables of \(R\) and \(S\) already
cover all variables of \(Q_\Delta\). %
Therefore, \(Q_\Delta\) cannot have more than \(N_R \times N_S\), and by symmetry, no
more than \(\min(N_RN_S, N_SN_T, N_RN_T)\). %
The work of Atserias, Grohe and Marx~\cite{atseriasSB2013} extends this notion
of covering all the variables to the idea of a \emph{fractional cover},
leading to an even better bound on the number of answers which is
\((N_RN_SN_T)^{1/2}\) and this bound is actually \emph{optimal} in the
sense that there exists an instance of \(Q_\Delta\) where \(R, S\) and
\(T\) have respectively sizes of at most \(N_R,N_S,N_T\) and
\(\bigot((N_RN_SN_T)^{1/2})\) answers where \(\bigot(\cdot)\) hides
polylogarithmic factors in the relation sizes and polynomial factors in the
query size, considered constant. %
Therefore, an algorithm able to compute the answers of \(Q_\Delta\) in time
\(\bigot((N_RN_SN_T)^{1/2})\) is a WCOJ algorithm for the class of
instances of \(Q_\Delta\) where \(R, S\) and \(T\) have sizes of at most
\(N_R,N_S,N_T\) respectively. %
It is optimal in the sense that it is linear in the size of the worst
possible instance of the class. %

In this simplified example, the class is defined via \emph{cardinality
  constraints}: we consider instances where each relation has a size (or
cardinality) that is bounded by a given integer. %
Building on the understanding of the worst case for such classes given
in~\cite{atseriasSB2013}, Ngo, Porat, R{\'{e}} and Rudra proposed the first
WCOJ algorithm for instances defined by cardinality constraints
in~\cite{NgoPRR12}. %
A simplified branch and bound algorithm, Triejoin, has been proposed by
Veldhuizen in~\cite{veldhuizen2014triejoin} and a more general version,
known as GenericJoin has been introduced by Ngo in~\cite{ngoWCOJ2018},
which is also worst case optimal on classes defined by so-called
\emph{acyclic degree constraints}, which is a strict generalisation. %
Since then, a fruitful line of research has focused in understanding worst
case bounds for classes of instances defined via more complex constraints
(e.g., functional dependencies or non acyclic degree constraints). %
A deep connection with information theory has been made
in~\cite{Khamis0S17} by Khamis, Ngo and Suciu, allowing the design of
PANDA, which can perform join queries in time that is not far from worst
case optimality, see~\cite{Suciu23} for an enlighting survey by Suciu on
this connection.

Another related line of research has focused on designing algorithms to
uniformly sample answers of join queries. %
One naive way of doing so is to first list \(\ans\) explicitly and then
uniformly sample an element of the list. %
Using a WCOJ algorithm, this gives a method allowing constant time sampling
after a preprocessing linear in the worst case. %
This complexity however does not match the intuition one could have of the
hardness of the problem. %
Indeed, it is reasonable to expect a query to be easier to sample if it has
many solutions, because, intuitively, they are easier to find. %
It turns out that this intuition can be turned into a formally proven
algorithm which achieves the following: for a class \(\C\) of queries
defined via cardinality constraints, Deng, Lu and Tao~\cite{dengJSH2023}
simultaneously with Kim, Ha, Flechter and Han~\cite{kimAGMOUT2023} proved
that one can achieve uniform sampling for a join query \(Q \in \C\) in time
of \(\bigot({|\wc(\C)| \over \max(1, |\ans|)})\) where \(\wc(\C)\) is the
worst case instance of class \(\C\) and \(\ans\) is the set of answers of
\(Q\). %
This result has recently been generalised to the case of acyclic degree
constraints by Wang and Tao~\cite{wangJSA2024}. %

\subparagraph*{Our contributions.} %
In this paper, we propose a very simple join algorithm with a very simple
analysis which achieves worst case optimality on classes of instances
defined by cardinality constraints and degree constraints. %
Our algorithm is a simple branch and bound algorithm which assigns one
variable to every possible value in the domain and backtracks whenever an
inconsistency is detected. %
As such, this can be seen as an extremely simplified version of
GenericJoin~\cite{ngoWCOJ2018} or TrieJoin~\cite{veldhuizen2014triejoin}. %
However, for these algorithms, a clever data structure is needed to branch
only on relevant values. %
This is actually necessary since this naive branch and bound algorithm is
not really worst case optimal. %
Indeed, on a query on domain \(D\), an extra factor of \(\size{D}\) appears
in the complexity. %
We turn it into a WCOJ algorithm with a simple trick: instead of branching
directly on domain values, we branch on the values bit by bit. %
An illustration of our algorithm is given in \cref{fig:example_dpll} for the
triangle query \(Q_\Delta\), defined in \cref{eq:q_delta} with tables given
in \cref{tab:example_db}. %
On the left, we show the branch and bound algorithm where values of
\(x_1,\dots,x_3\) are iteratively set to values in the domain
\(\{0,1,2,3\}\). %
Whenever a relation is inconsistent with the current partial assignment, we
backtrack, which is represented by \(\bot\) in the tree. %
Observe on the example that after setting \(x_1\) to \(0\), we explore many
``useless'' values for \(x_2\) that directly give an inconsistency. %
In this simple example, we can directly read from \(R\) that only the value
\(x_2 \mapsto 0\) is relevant, but in some more complex queries, we may
need to compute more complex intersections efficiently, which is exactly
how GenericJoin and TrieJoin address the problem. %

To avoid the need for such a data structure, we slightly modify the
algorithm as shown on the right part of \cref{fig:example_dpll}. %
We encode the domain \(\{0,1,2,3\}\) with two bits on \(\{0,1\}^2\) and now
branch on the first bit \(x_1^1\) of \(x_1\) and then on the second bit
\(x_1^2\) of \(x_1\) and so on. %
We can directly see on the example that when \(x_1\) is set to \(0\), that
is, when \(x_1^1 \mapsto 0, x_1^2 \mapsto 0\), then we do not explore the
values \(\{1,3\}\) for \(x_2\) as we directly detect an inconsistency when
setting the first bit of \(x_2\) to \(1\). %
This simple trick is enough to guarantee worst case optimality of a simple
branch and bound algorithm. %

\begin{table}[htp]
  \centering
  \begin{tabular}{>{\sffamily\bfseries}c|>{\sffamily}c>{\sffamily}cm{1.5cm}>{\sffamily\bfseries}c|>{\sffamily}c>{\sffamily}cm{1.5cm}>{\sffamily\bfseries}c|>{\sffamily}c>{\sffamily}c}
R & \(\mathsf{x_1}\) & \(\mathsf{x_2}\) & &    S & \(\mathsf{x_2}\) & \(\mathsf{x_3}\) & & T & \(\mathsf{x_1}\) & \(\mathsf{x_3}\)   \\ \cline{1-3} \cline{5-7} \cline{9-11}
  & 0   & 0   & &      & 0   & 2   & &   & 0   & 3   \\
  & 1   & 0   & &      & 0   & 3   & &   & 1   & 0   \\
  & 1   & 1   & &      & 1   & 0   & &   & 1   & 2   \\
  & 2   & 1   & &      & 1   & 2   & &   & 2   & 3   \\
  \end{tabular}
  \caption{An instance of \(Q_\Delta\) on domain \(\{0,1,2,3\}\).}
  \label{tab:example_db}
\end{table}

\begin{figure}[htp]
  \centering
  \begin{subfigure}[b]{8cm}
    \includegraphics[width=8cm]{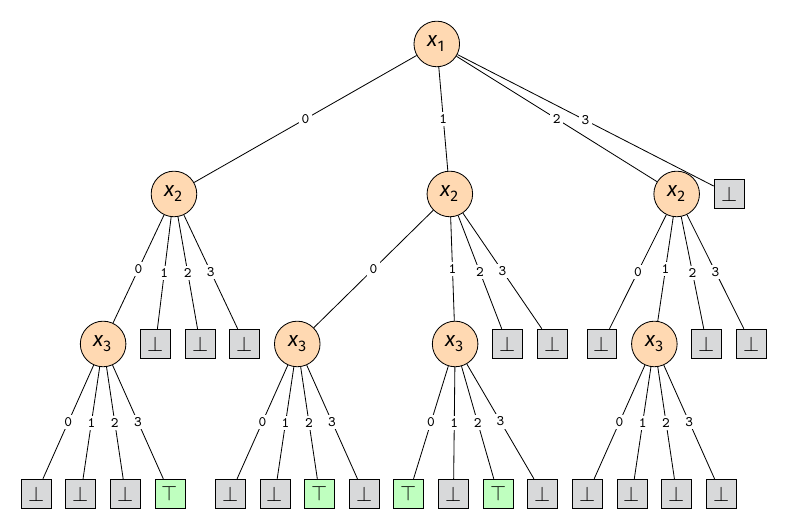}
    \caption{Over \(4\)-valued domain}
  \end{subfigure}
  \hfill
  \begin{subfigure}[b]{5cm}
    \includegraphics[width=5cm]{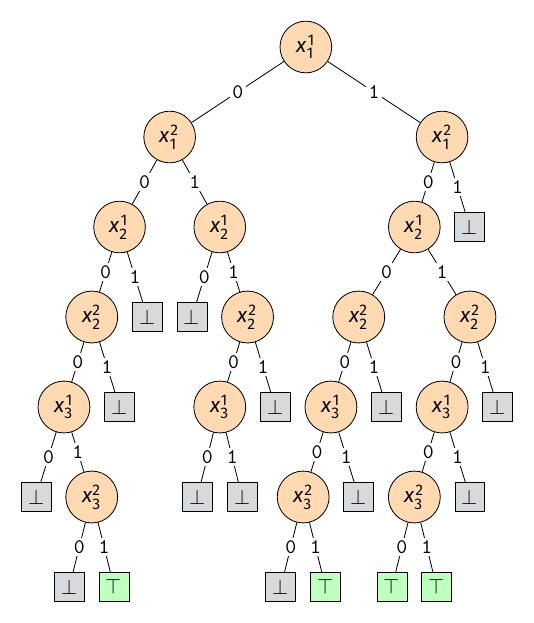}
    \caption{Binarised version}
  \end{subfigure}
  \caption{Trace of an execution of our algorithm over the triangle query
    \(Q_\Delta\) defined in \cref{tab:example_db}. Satisfying assignments
    are labelled \(\top\) and shown with a greenish node.On the right side,
    the same algorithm branching on the bits of the values instead of the
    values themselves.}
  \label{fig:example_dpll}
\end{figure}

One strength of our result is that worst case optimality is proven without
any knowledge of the actual worst case of the class. %
This is in contrast with most existing WCOJ algorithms, whose analysis
often relies on how the worst case value is computed. %
This knowledge is infused into the proof and, sometimes even, in the
algorithm itself, for example in the first NPPR
algorithm~\cite{NgoPRR12}. %
Our analysis only exploits one property that we call \kl{\emph{prefix
    closedness}}. %
A class of instances \(\C\) is intuitively \kl{prefix closed} if for every
\(Q \in \C\), the number of answers of \(Q\) where we have removed some
variables never exceeds the worst case \(\wc(\C)\). %
This property is straightforward to establish for classes defined with
cardinality constraints and with acyclic degree constraints resulting in an
elementary proof of worst case optimality. %

The second contribution of this paper is to show how uniformly sampling
answers can be achieved in expected runtime of \(\bigot({\wc(\C) \over
  \max(1,|\ans|)})\) for classes defined with cardinality constraints and
with acyclic degree constraints, matching the complexity established in
previous work with more involved
techniques~\cite{wangJSA2024,dengJSH2023,kimAGMOUT2023}. %
Our approach is elementary. %
Intuitively, we see the trace of our WCOJ algorithm as a tree whose leaves
are either conflicts or solutions. %
The sampling problem hence reduces to uniformly sampling ``interesting''
leaves in a tree, without fully exploring it. %
It turns out that this is easy to do by adapting an algorithm from
Rosenbaum~\cite{rosenbaumSLTEP1993} as long as one has a way of
overestimating the number of interesting leaves in each subtree. %
We show that this can be done for join queries using the knowledge we have
on how to compute worst case bounds. %
The only technical blackbox we use to establish this result is (a weak form
of) Friedgut's inequality~\cite{friedgut04}. %
In particular, we recover the recent result from~\cite{wangJSA2024} on
sampling join queries under acyclic degree constraints with an elementary
proof. %

\subparagraph*{Organisation of the paper.} %
We give some necessary notations and preliminaries in
\cref{sec:preliminaries}. %
\cref{sec:dpll} contains the description of our branch and bound algorithm
and a simple analysis of its complexity. %
We then show that this is enough to establish worst case optimality for
classes defined with cardinality constraints and acyclic degree constraints
in \cref{sec:wc}. %
Finally, \cref{sec:sampling} shows that the branch and bound algorithm can
easily be turned into a sampling algorithm achieving the same complexity as
previous work in a simpler and more modular way. %


\section{Preliminaries}
\label{sec:preliminaries}

\subparagraph{Notations.}
\label{ssec:notations}

We assume the reader familiar with the basic vocabulary of database theory
and mostly introduce notations in this section. %
Given two sets \(X\) and \(D\), we denote by \(D^X\) the set of
\AP\intro{tuples} over variables \(X\) and domain \(D\), that is, the set
of mappings from \(X\) to \(D\). %
We denote by \(\emptytuple\) the empty tuple, that is, the only element of
\(D^\emptyset\) and by \(\unittuple\) the tuple on variable \(\{x\}\) that
maps \(x\) to \(d\). %
For \(\tau \in D^X\) and \(\sigma \in D^Y\) with \(X \cap Y = \emptyset\), we denote by \(\tau \cup
\sigma\) the tuple mapping \(z \in X \cup Y\) to \(\tau(z)\) if \(z \in X\)
and \(\sigma(z)\) otherwise. %

A \AP\intro{relation} \(R\) is a subset of \(D^X\). %
Given \(\tau \in D^X\) and \(Y \subseteq X\), we denote by \(\tau_{|Y}\) the restriction of
\(\tau\) to \(Y\), that is, the tuple such that \(\tau_{|Y}(y) = \tau(y)\)
for every \(y \in Y\). %
For \(R \subseteq D^X\), we write \(R_{|Y}\) for \(\{\tau_{|Y} \mid \tau \in
R\}\). %
From now on when mentioning a relation, for example \(R\), we assume that
\(X_R\) is the set of variables on which it is defined, i.e. \(R\subseteq
D^{X_R}\). %
Let \(\tau \in D^Y\), we denote by \(R[\tau]\) the relation \(\{\sigma_{|X_R - Y} \mid \sigma \in
R,\, \sigma_{|Y} = \tau_{|X_R}\}\). %
That is \(R[\tau]\) is obtained by filtering out every \kl{tuple} of \(R\)
that do not agree with \(\tau\) on the variables \(Y\). %

A \AP\intro{join query} \(Q\) over variables \(X\) and domain \(D\) is a
set of relations such that, for every \(R \in Q\), \(R \subseteq D^{X_R}\)
for some \(X_R \subseteq X\). %
Observe that, as it is often done in the literature about WCOJ algorithms,
we slightly deviate from the usual database setting which separates the
data from the query. %
We can still see a join query as a usual full conjunctive query with
hypergraph \((X, \{X_R \mid R \in Q\})\) and the data, that is, the tuples
of each \(R \in Q\). %

The \AP\intro{answer set of \(Q\)}, denoted by \(\ans\), is defined as the
set of tuples \(\tau \in D^X\) such that for every \(R \in Q\),
\(\tau_{|X_R} \in R\). %
The \AP\intro{join problem} is the problem of outputing \(\ans\) given
\(Q\) as input. %
The \AP\intro{data size} of \(Q\), denoted as \(\dsize\), is defined as the
number of tuples in its relations, \(\dsize = \sum_{R \in Q} |R|\). %
Given \(Y \subseteq X\), we denote by \(Q_{|Y}\) the join query defined as \(\{R_{|Y
  \cap X_R} \mid R \in Q\}\). %
Given \(\tau \in D^Y\), we define \(Q[\tau] = \{R[\tau]\mid R \in Q\}\). %
We say that \(\tau\) is \AP\intro{inconsistent with \(Q\)} if \(Q[\tau]\)
contains an empty relation; \(\tau\) is otherwise \AP\intro{consistent with
  \(Q\)}. %
We make the following observation that will be crucial for the rest of this
paper: %

\begin{lemma}
  \label{lem:incpref}
  For every \(\tau \in D^Y\), \(\tau \in \ans[Q_{|Y}]\) iff \(\tau\) is consistent with
  \(Q\).
\end{lemma}
\begin{proof}
  It is simply a reformulation: if \(\tau \in \ans[Q_{|Y}]\) then it means that
  for every \(R \in Q\), \(\tau_{|X_R \cap Y} \in R_{|Y}\). %
  In particular, \(R[\tau]\) is not empty. %
  Hence \(\tau\) is consistent with \(Q\). %
  Conversely, if \(\tau\) is \kl{consistent} with \(Q\), then for every \(R \in
  Q\), \(R[\tau]\) is not empty. %
  That is, there exists some tuple \(\sigma \in R\) such that \(\sigma_{|Y} = \tau_{|X_R \cap
    Y}\). %
  In other words, \(\tau_{|X_R \cap Y} \in R_{|Y}\) for every \(R \in Q\), hence,
  \(\tau \in \ans[Q_{|Y}]\).
\end{proof}

In this paper, we will always make the assumption that the domain \(D\) of
a join query is its \emph{active domain}, that is, the set of values that
appear in at least one relation. %
Moreover, we assume that every value in this active domain is encoded with
\(\bigo(\log\size{D})\) bits. %
While this is a reasonable assumption, it may not be completely realistic
in practice (for example when using string values). %
We can still enforce this condition with linear preprocessing by reencoding
the domain using a perfect hash function~\cite{czech1997perfect}. %

\subparagraph{Worst-case optimal join.}
\label{ssec:wcoj}

In this section, we give an abstract definition of what we call a \WCOJ{}
algorithm. %
Let \(H = (X,E)\) be a hypergraph and \(\C\) be a class of join queries
with hypergraph \(H\). %
We define the \AP\intro{worst case of \(\C\)}, denoted by \(\wc(\C)\) as
\wc\((\C) = \sup_{Q \in \C} \size{\ans}\). %
An algorithm is a \emph{\WCOJ{} for \(\C\)} if, on input \(Q \in \C\), it outputs
\(\ans\) in time \(\bigot(\wc(\C) \times \poly(n,m))\), where
\(\bigot(\cdot)\) hides polylog factors, \(m = \size{E}\) and \(n =
\size{X}\) are parameters that only depend on the structure of the query
and not on the content of the relations. %
Of course, for this definition to make sense, one needs \(\wc(\C)\) to be
finite. %
Many such classes have been studied in the literature and many \WCOJ{}
algorithms have been proposed. %
In this paper, we will focus on the two main classes that have been
considered: classes defined via \kl{cardinality constraints} and classes
defined via \kl{degree constraints}.

\subparagraph{Cardinality Constraints.} %

Let \(H=(X,E)\) be a hypergraph verifying \(\bigcup E = X\) (every node is
covered by a hyperedge) and let \(\tup[N] \in \N^E\). %
We let \(\CCard\) be the class of join queries \(Q\) on hypergraph \(H\)
such that for every \(e \in E\), there is \(R_e\) in \(Q\) such that
\(X_{R_e} = e\) and \(\size{R_e} \leq \tup[N](e)\). %
We say that \(\CCard\) is a class defined via \AP\intro{cardinality
  constraints} because it puts a bound on the cardinality of (the
intersection of) the involved relations. %
Clearly, \(\wc(\CCard) \leq \prod_{e \in E} \tup[N](e) < \infty\). %
Actually, one can get a sharper, almost optimal upper bound on
\(\wc(\CCard)\) using a result by Grohe and Marx~\cite{groheCSF2014}
(optimality was proven by Atserias, Grohe and Marx in
\cite{atseriasSB2013}) and which has later been known as the AGM bound. %
For example, one can show that for the triangle hypergraph \(H_\Delta =
\{e_1,e_2,e_3\}\) where \(e_1 = \{1,2\}\), \(e_2 = \{2,3\}\) and \(e_3 =
\{1,3\}\), \(\wc(\CCard[\tup[N]][H_\Delta]) \leq
\sqrt{N(e_1)N(e_2)N(e_3)}\). %
We delay the precise presentation of such bounds to \cref{sec:sampling}
where we are interested in sampling answer from conjunctive queries. %
One strength of our worst case optimal join approach compared to previous
work is that we do not need to have an understanding of the worst case
bound to prove its worst case optimality. %

\subparagraph{Degree constraints.}

Another class of join queries which received attention in the literature on
\WCOJ{}s is the class of queries defined with degree constraints. %
Given two sets \(A \subseteq B\), a \intro{degree constraint} is a triplet of the
form \((A,B,N_{B|A})\) with \(N_{B|A} > 0\). %
A relation \(R_e\) on variables \(e \supseteq B\) respects the degree constraint
\((A,B,N_{B|A})\) if and only if \(\max_{\tau \in D^A} |R[\tau]_{|B}| \leq
N_{B|A}\). %
It is a generalisation of cardinality constraints since a cardinality
constraint can be seen as a degree constraint of the form \((\emptyset, B,
N_B)\). %
It can also be seen as a generalisation of functional dependencies since a
functional dependency \(X \rightarrow y\) can be seen as the \((X, \{y\},
1)\) \kl{degree constraint}. %
Let \(H=(X,E)\) be a hypergraph and \(DC\) be a set of degree constraints
of the form \((A,B,N_{B|A})\) with \(A \subseteq B \subseteq X\). %
Each degree constraint \(\delta \in DC\) is associated with an hyperedge
\(e_\delta \in E\) with \(e_\delta \supseteq B\) which guards it. %
We let \(\CDeg\) be the class of queries \(Q\) on hypergraph \(H\) such
that for every \(\delta=(A,B,N_{B|A}) \in DC\), there is an atom \(R\) in
\(Q\) such that \(X_R = e_\delta\) and \(R\) respects \(\delta\). %

Observe that it may happen that \(\wc(\CDeg) = +\infty\). %
In this paper, we are only interested in classes where this does not
happen. %
This is often enforced by assuming that \(\bigcup E = X\) and that for every \(e
\in H\), at least one constraint in \(DC\) is a cardinality constraint of
the form \((\emptyset, e, N_e)\) with guard \(e\) that has hence to be
respected by a relation \(R\) with \(X_R = e\). %
In this case, as before, \(\wc(\CDeg) \leq \prod_e N_e < +\infty\). %
Here again, more precise upper bounds are known on \(\wc(\CDeg)\) but they
will not be necessary for our worst case optimal join algorithm and we
delay this discussion to \cref{sec:sampling} where we will need them. %


\section{Branch and bound algorithm for join queries}
\label{sec:dpll}

In this section, we propose a simple branch and bound algorithm to compute
join queries and provide an easy upper bound on its complexity. %
We will show later how this upper bound can be proved to be worst case
optimal for some classes of instances. %
The algorithm can be seen as an instance of GenericJoin from~\cite{NgoRR13} but
it is given in an extremely simple form and its analysis is elementary. %
Written in this way, the algorithm is not worst case optimal but a simple
algorithmic trick will allow us to recover known results, presented in
\cref{sec:wc}. %

The algorithm, whose pseudo code is given in \cref{alg:wcj}, is a simple
recursive search: assume a fixed order \((x_1,\dots,x_n)\) is given on
variables \(X\). %
We find the answers of \(Q\) by setting variables sequentially according to
this order, trying each possible value in the domain. %
Whenever the current partial assignment is inconsistent with \(Q\), it is
not further expanded. %
If every variable is assigned and the assignment is consistent with \(Q\),
then it is output. %

\begin{algorithm}[htp]
  \caption[]{An algorithm to compute join queries}
  \label{alg:wcj}
  \begin{algorithmic}[1]
    \Procedure{\(\WCJ\)}{$Q,\tau$}
    \eIf{\(Q[\tau]\) contains an empty relation}{\Return}
    \State \(i \gets \) last variable assigned by \(\tau\);
    \eIf{\(i = n\)}{output \(\tau\). \Return}
    \eFor{\(d \in D\)}{\(\WCJ(Q, \tau \cup \unittuple[x_{i+1}][d])\)}
    \EndProcedure
  \end{algorithmic}
\end{algorithm}

\subparagraph*{Correction of the algorithm.}

Starting with a call \(\WCJ(Q, \emptytuple)\), every recursive call is of
the form \(\WCJ(Q, \tau)\) where \(\tau\) is a tuple in \(D^{X_i}\) where
\(X_i := \{x_1,\dots,x_i\}\). %
We claim that for every \(\tau\) which assigns variables \(X_i\), then
\(\WCJ(Q,\tau)\) outputs \(\tau \cup \sigma\) for every answer \(\sigma\)
of \(Q[\tau]\). %
The proof is by induction on \(i\). %
If \(i=n\), then \(\tau\) is output if and only if \(Q[\tau]\) does not contain
the empty relation, which by \cref{lem:incpref} means that \(\tau\) is an
answer of \(Q\). %
Now assume \(i < n\). %
If \(\tau\) is \kl{inconsistent} with \(Q\) then nothing is output, this is
coherent with our induction hypothesis since \(Q[\tau]\) contains an empty
relation, meaning that any tuple \(\sigma\) so that \(\sigma_{|X_i} =
\tau\) is not in \(\ans\). %
Otherwise, by induction, \(\WCJ(Q,\tau \cup \unittuple[x_{i+1}][d])\) outputs
\(\tau \cup \unittuple[x_{i+1}][d] \cup \sigma\) for every \(\sigma
\in\ans[Q[\tau\cup \unittuple[x_{i+1}][d]]]\), that is, for every \(\sigma
\in \ans[Q[\tau]]\). %
It completes the induction and it directly follows that
\(\WCJ(Q,\emptytuple)\) outputs \(\ans[Q]\). %

\subparagraph*{Number of recursive calls.}

We claim that \cref{alg:wcj} does at most \((1 + \size{D}) \cdot \sum_{i \leq n}
\size{\ans[Q_{|X_i}]}\) recursive calls. %
Indeed, as stated before, every recursive call is of the form \(\WCJ(Q,\tau)\)
where \(\tau\) is a tuple of \(D^{X_i}\). %
In the first case, assume that \(Q\) is \kl{consistent} with \(\tau\), which
means in particular that \(\tau\) is in \(\ans[Q_{|X_i}]\) by
\cref{lem:incpref}. %
Hence, there are at most \(\sum_{i \leq n} \size{\ans[Q_{|X_i}]}\) recursive calls
of this type. %
In the second case, assume that \(\tau\) is \kl{inconsistent} with \(Q\). %
Then the recursive call with parameters \((Q,\tau)\) has been issued from a
call of the form \((Q, \tau')\) where \(\tau = \tau' \cup
\unittuple[x_i][d]\) for some \(d \in D\). %
In particular, \(\tau'\) is consistent with \(Q\), otherwise such a recursive
call would not have happened. %
Hence, \(\tau' \in \ans[Q_{|X_{i-1}}]\) and there are at most \(\size{D}\)
possible \(\tau\) for a given \(\tau' \in \ans[Q_{|X_{i-1}}]\). %
Therefore, there are at most \(\size{D} \cdot \sum_{i \leq n} \size{\ans[Q_{|X_i}]}\)
recursive calls of this form, this in total, \((\size{D} + 1)\sum_{i \leq
  n} \size{\ans[Q_{|X_i}]}\) recursive calls. %

\subparagraph*{Efficient implementation.}

Now we explain how, using a very simple data structure, one can assume that
each recursive call is executed in \(\bigot(m)\) where \(m\) is the number
of atoms in \(Q\). %
The only non trivial thing is to check whether \(Q[\tau]\) contains an empty
relation. %
To do that, we simply assume that every relation is given sorted in
lexicographical order, for the attribute order \(x_1, \dots, x_n\). %
This could be obtained via a preprocessing that is quasi linear in the data
(or linear in the RAM model, but since we ignore polylogarithmic factors,
it does not matter much). %
Now observe that if \(R\) is a relation of \(Q\) and \(\tau\) a tuple in
\(D^{X_i}\), then all tuples from \(R[\tau]\) are consecutively stored in
the table. %
Hence we can represent \(R[\tau]\) by keeping two pointers \(p_1,p_2\) on the
tuples of \(R\): one towards the first tuple and one towards the last tuple
in \(R[\tau]\). %
To check whether \(R[\tau]\) is \kl{consistent}, it is enough to check that
\(p_1 \leq p_2\). %
To go from the representation of \(R[\tau]\) to the representation of \(R[\tau
\cup \unittuple[x_{i+1}][d]]\), we simply need to find the first and last
tuple between \(p_1\) and \(p_2\) where \(x_{i+1} = d\). %
This can be done via a binary search in time \(\bigo(\log\size{R})\). %
Hence, each recursive join can be executed in time \(\bigo(m \log\dsize\),
that is, \(\bigot(m)\). %
A slightly more involved data structure would allow us to compute in time
\(\bigo(m)\) by representing \(R\) as a trie as
in~\cite{veldhuizen2014triejoin}. %
We just proved: %

\begin{theorem}
  \label{thm:dpll_complexity}
  Given a join query \(Q\) on domain \(D\) with \(m\) atoms and
  \((x_1,\dots,x_n)\) an order on the variables of \(Q\),
  \(\WCJ(Q,\emptytuple)\) computes \(\ans[Q]\) in time \(\bigot(m|D| \cdot
  \sum_{i \leq n} |\ans[Q_{|X_i}]|)\), where \(X_i=\{x_1,\dots,x_i\}\). %
\end{theorem}


\section{Worstcase optimality}
\label{sec:wc}

\subsection{Prefix closed classes}

To show that \cref{alg:wcj} is worst case optimal on a class \(\C\) of
instances, we need to bound \cref{thm:dpll_complexity} by
\(\bigot(\wc(\C))\). %
Of course, this will not be true for any class of instances but it turns
out that we can easily do so on classes defined by cardinality constraints
or by acyclic degree constaints. %
\cref{thm:dpll_complexity} motivates the following definition: a class
\(\C\) is \AP\intro{prefix closed for the order \(\pi = (x_1,\dots, x_n)\)}
if and only if for every \(i \leq n\) and \(Q \in \C\),
\(\size{\ans[Q_{|X_i}]} \leq \wc(\C)\). %
Indeed, if \(\C\) is \kl{prefix closed} for an order \(\pi\), then computing
\(\ans\) for \(Q \in \C\) using \cref{alg:wcj} with order \(\pi\) will take
\(\bigot(mn \cdot \size{D} \cdot \wc(\C))\), where \(D\) is the domain of
\(Q\). %

\begin{theorem}
  \label{thm:wcjcomplexity}
  For every class \(\C\) that is \kl{prefix closed} for an order
  \((x_1,\dots,x_n)\) and join query \(Q \in \C\) with \(n\) variables and \(m\)
  relations, \(\WCJ(Q,\emptytuple)\) returns \(\ans\) in time \(\bigot(nm
  \cdot \size{D} \cdot \wc(\C))\). %
\end{theorem}

While \(m\) and \(n\) are considered constant in our setting, we cannot
assume so for \(\size{D}\). %
Hence, \cref{thm:dpll_complexity} and \kl{prefix closedness} will not be
enough to establish worst case optimality of \cref{alg:wcj}. %
That being said, we present a simple trick in \cref{sec:binarisation} which
allows us to circumvent this issue easily. %
The main classes for which worst case optimal algorithms are known are
prefix closed, at least for one order. %
Even if cardinality constraints are less general than degree constraints,
we start by showing it for the former as a warmup, even if the proof is
essentially the same for the latter: %

\begin{theorem}
  \label{thm:card-prefix-closed}
  Let \(\CCard\) be a class of join queries defined for hypergraph
  \(H=(X,E)\) and cardinality constraints \(\tup[N] \subseteq \N^E\). %
  Then \(\CCard\) is \kl{prefix closed} for every order. %
\end{theorem}
\begin{proof}
  Let \(Q \in \CCard\), \((x_1,\dots,x_n)\) be an order on \(X\) and \(i \leq
  n\). %
  We need to show that \(\ans[Q_{|X_i}] \leq \wc(\CCard)\). %
  To do so, we construct \(Q^* \in \CCard\) such that \(\size{\ans[Q_{|X_i}]} =
  \size{\ans[Q^*]}\). %
  Since \(Q^* \in \CCard\), we have by definition that \(\ans[Q^*] \leq
  \wc(\CCard)\), hence \(\ans[Q_{|X_i}] \leq \wc(\CCard)\). %

  Assume that \(Q\) is on domain \(D \neq \emptyset\) and let \(d \in D\) be some fixed
  element of \(D\). %
  We denote by \(d^Y \in D^{Y}\) the tuple defined as \(d^Y(y) = d\) for every
  \(y \in Y\). %
  Let \(R \in Q_{|X_i}\). %
  By definition, \(R = R_e\vert_{X_i}\) for some \(e \in E\). %
  Hence, \(\size{R} \leq \size{R_e} \leq \tup[N](e)\). %
  We define \(R^*_e \subseteq D^e\) as \(R \times \{d^{e \setminus X_i}\}\), that is, we extend every
  tuple from \(R\) to variables \(e\) by setting every missing variable to
  \(d\). %
  Clearly, \(\size{R^*} = \size{R} \leq \size{R_e} \leq \tup[N](e)\). %
  Hence the query \(Q^*\) defined as \(\{R^*\mid R \in Q_{|X_i}\}\) is in
  \(\CCard\). %
  Moreover, we clearly have \(\ans[Q^*] = \ans[Q_{|X_i}] \times \{d^{X \setminus X_i}\}\),
  therefore \(\size{\ans[Q^*]} = \size{\ans[Q_{|X_i}]}\) as needed to
  complete the proof. %
\end{proof}

We now generalise the previous result to classes defined via degree
constraints. %
Observe however that such classes may not always be \kl{prefix closed}, or
sometimes only for some particular order. %
For example, consider the query \(Q = R(x_3,x_1) \wedge S(x_3,x_2)\) and consider
the class \(\C\) respecting functional dependencies \(x_3 \rightarrow x_1\)
and \(x_3 \rightarrow x_2\) and cardinality constraints \(\size{R} \leq N\)
and \(\size{S} \leq N\). %
Clearly, \(\wc(\C_Q) \leq N\) since once \(x_3\) is fixed, so are \(x_1\) and
\(x_2\). %
Now, consider an instance \(Q^*\) where \(R^{*} = S^{*} = \{(i,i) \mid 0 < i \leq
N\}\). %
It is easy to see that \(Q^* \in \C\) and that \(Q^*_{|\{x_1,x_2\}}\) has \(N^2 >
\wc(\C_Q)\) answers. %
The previous example is not \kl{prefix closed} for \((x_1,x_2,x_3)\) because
we chose an order that goes in the wrong direction in regard to the
functional dependencies. %
One can check that \(\C_Q\) is \kl{prefix closed} for the order
\((x_3,x_2,x_1)\). %

This motivates the following definition: for \(H=(X,E)\) a hypergraph and
\(DC\) a set of degree constraints, we define the \emph{dependency graph}
\(G_{DC}\) as the graph whose vertex set is \(X\) and where there is an
edge \(u \rightarrow v\) if and only if there is a degree constraint
\((A,B,N_{B|A})\) in \(DC\) with \(u \in A\) and \(v \in B\). %
We say that \(DC\) is acyclic if \(G_{DC}\) is acyclic. %
In this case, an order \((x_1,\dots,x_n)\) is said to be \emph{compatible} with
\(DC\) if this is a topological sort of \(G_{DC}\). %
Unsurprisingly, this allows to prove the following generalisation of
\cref{thm:card-prefix-closed}: %

\begin{theorem}
  \label{thm:dc-prefix-closed}
  Let \(\CDeg\) be a class of join queries defined for hypergraph
  \(H=(X,E)\) and acyclic degree constraints \(DC\). %
  Then \(\CDeg\) is \kl{prefix closed} for every order compatible with
  \(DC\). %
\end{theorem}
\begin{proof}
  The proof is very similar to the proof of
  \cref{thm:card-prefix-closed}. %
  Let \(Q \in \CDeg\) and \(i \leq n\). %
  We construct \(Q^*\) as in \cref{thm:card-prefix-closed}. %
  We still have \(\size{\ans[Q_{|X_i}]} = \size{\ans[Q^*]}\). %
  We only have to check that \(Q^* \in \CDeg\). %
  Let \(\delta = (A,B,N_{B|A}) \in DC\) be a cardinality constraint. %
  By definition, it is respected by an atom \(R\) of \(Q\) on variables \(e
  \supseteq B\). %
  We claim that \(R^* \in Q^*\) also respects \(\delta\). %
  Indeed \(X_i \cap e \subseteq A\), then for every \(\tau \in D^A\), there is at most one
  tuple in \(R^*[\tau]\) which is \(\tau \times d^{e \setminus X_i}\),
  hence \(\size{R^*[\tau]_{|B}} \leq 1 \leq N_{B|A}\). %
  Otherwise, since the order is compatible with \(DC\), \(A \subseteq
  X_i\). %
  Hence \(R^*[\tau] = R_{|X_i}[\tau] \times d^{e \setminus X_i}\). %
  In particular \(\size{R^*[\tau]} = \size{R_{|X_i}[\tau]} \leq
  \size{R[\tau]}\). %
  Hence projecting out on \(B\), \(\size{R_{|X_i}[\tau]_{|B}} \leq
  \size{R[\tau]_{|B}} \leq N_{B|A}\) since \(R\) respects the degree
  constraint \((A,B,N_{B|A})\). %
  Hence, \(R^*\) also respects this degree constraint. %
  Since this reasoning works for every \(R^* \in Q^*\), we conclude that \(Q^*
  \in \CDeg\). %
  Hence \(\size{\ans[Q_{|X_i}]} = \size{\ans[Q^*]} \leq \wc(\CDeg)\), which is
  what we needed to prove. %
\end{proof}

A direct corollary of \cref{thm:wcjcomplexity,thm:dc-prefix-closed} is that
\cref{alg:wcj} is almost worst case optimal on classes defined by acyclic
degree constraints. %

\begin{corollary}
  \label{cor:almostwc}
  Let \(\CDeg\) be a class of join queries defined for hypergraph
  \(H=(X,E)\), \(m=\size{E}, n=\size{X}\) and acyclic degree constraints
  \(DC\). %
  Assume \((x_1,\dots,x_n)\) is an order compatible with \(DC\). %
  Then for every \(Q \in \CDeg\), \(\WCJ(Q,\emptytuple)\) returns \(\ans\) in
  time \(\bigot(mn \cdot \size{D} \cdot \wc(\CDeg))\). %
\end{corollary}

Observe that in order to prove worst case optimality of \cref{alg:wcj} in
\cref{cor:almostwc}, we have not used any knowledge on the actual value of
\(\wc(\CDeg)\), which makes our approach simpler than existing analysis of
worst case optimal join algorithms. %

\subsection{Binarisation}
\label{sec:binarisation}

We have seen that \cref{alg:wcj} achieves \(\bigot(mn\size{D} \cdot \wc(\C))\)
complexity when \(\C\) is \kl{prefix closed}, which does not qualify as a
worst case optimal join yet. %
The extra \(\size{D}\) factor comes from the fact that we are testing every
possible value of \(d \in D\) for each variable, even if many of them will
directly lead to inconsistencies. %
We could overcome this issue by exploring only relevant values, using for
example the trie join algorithm from~\cite{veldhuizen2014triejoin} which
allows to enumerate values present in the intersection of every relation in
time \(\bigo(\log\size{D})\) or Hash indices as in~\cite{ngo2018worst}. %
While these techniques are interesting for practical implementation, our
goal in this paper is to use as little technical tools as possible. %
Hence, we present here a new simple technique to remove the extra
\(\size{D}\) factor. %
The main idea is that instead of testing every value in the domain for each
variable, we fix its value bit by bit. %
This could be implemented directly by modifying \cref{alg:wcj} or, as we
chose to present it, by transforming any join query \(Q\) on domain \(D\)
with \(n\) variables into a join query \(\bin[Q]\) with \(n \cdot b\)
variables where \(b = \lceil\log\size{D}\rceil\) variables on domain
\(\{0,1\}\) such that the answers of \(\bin[Q]\) are in one-to-one
correspondence with the answers of \(Q\). %
We do this by reencoding each element of the domain \(D\) in binary. %

More formally, let \(Q\) be a \kl{join query} on variables \(X\) and domain
\(D\). %
Without loss of generality, we assume that \(D = \{1,\dots,d\}\) for some \(d\)
and we let \(b = \lceil \log\ d \rceil\) to be the number of bits needed to
encode every element of \(D\). %
We represent each element \(k\) in \(D\) by the binary number \(\bin[k]\)
representing \(k\) and written with \(b\) bits. %
For \(1 \leq i \leq b\), let \({\bin[k]}\lbrack i\rbrack\) be the \(i^\mathsf{th}\) bit of a binary
representation of \(k \in D\). %
The function \(\bin[\cdot]\) is a bijection between \(D\) and its image. %

We now lift the functions \(\bin[\cdot]\) %
to pairs of bijections over \kl{tuples}, \kl{relations} and then over
\kl{join queries}. %
For a set of variables \(Y\), we denote by \(\bin[Y]\) the set \(\{y^i \mid y
\in Y, 1 \leq i \leq b\}\), that is, the set containing \(b\) distinct
copies of each variable of \(Y\). %
For \(\tau\in D^Y\), we define \(\bin[\tau]\) as follows: for every \(y\in Y\) and
\(i\in[b]\), \(\bin[\tau](y^i) = \bin[\tau(y)]\lbrack i\rbrack\). %
Given a relation \(R\subseteq D^Y\) we let \(\bin[R] = \{\bin[\tau] \mid \tau \in
R\}\). %
Finally, given a join query \(Q\) over variables \(X\) and domain \(D\), we
let \(\bin[Q] = \{\bin[R]\mid R \in Q\}\). %
Obviously, the answers of \(\bin[Q]\) are in one-to-one correspondence with
the answers of \(Q\). %
Moreover, we have that the cardinalities of the relations are invariant
under this transformation, i.e. \(\size{\bin[R]} = \size{R}\). %
Applying \cref{thm:dpll_complexity} on \(\bin[Q]\) directly yields the
following: %

\begin{theorem}
  \label{thm:wcjcomplexity_bin}
  Given a join query \(Q\) on domain \(D \subseteq [2^b]\) with \(m\) atoms,
  \((x_1,\dots,x_n)\) an order on the variables of \(Q\),
  \(\WCJ(\bin[Q],\emptytuple)\) with order \((x_1^1,\dots,x_1^b, \dots, x_n^1, \dots, x_n^b)\)
  computes \(\ans[Q]\) in time \(\bigot(m \sum_{i \leq n}\sum_{j \leq b}
  \size{\ans[\bin[Q]_{|X^j_i}]})\) where \(X^j_i=\{x_1^1, \dots, x_1^b, \dots, x_i^1, \dots,
  x_i^j\}\). %
\end{theorem}

To show worst case optimality, it remains to bound
\(\max_{i,j}\size{\ans[\bin[Q]_{|X^j_i}]}\) by \(\wc(\C)\). %
We do this by showing that in the case of acyclic degree constraint,
\(\bin[Q]\) belongs to a class \(\bin[\C]\) defined by acyclic degree
constraints where \(\wc(\bin[\C]) \leq \wc(\C)\) and such that if \(\C\) is
\kl{prefix closed} for \(x_1,\dots,x_n\) then \(\bin[\C]\) is \kl{prefix
  closed} for \(x_1^1,\dots,x_1^b, \dots, x_n^1, \dots, x_n^b\). %
The idea is to binarise the degree constraints as follows: for \(b \in \N\)
and a degree constraint \(\delta = (A,B,N)\), we denote by \(\bin[\delta]\)
the degree constraint \((\bin[A],\bin[B],N)\) and for a set \(DC\) of
degree constraints, let \(\bin[DC] := \{\bin[\delta] \mid \delta \in
DC\}\). %
We show: %

\begin{lemma}
  \label{lem:binariseDC}
  Let \(DC\) be a set of degree constraints, \(H\) a hypergraph and \(b \in
  \N\). %
  For every \(Q \in \CDeg\) on domain \(D \subseteq [2^b]\), we have \(\bin[Q] \in
  \CDeg[\bin[DC]][\bin[H]]\). %
  Moreover, \(\wc(\CDeg[\bin[DC]][\bin[H]]) \leq \wc(\bin[Q])\). %
  Finally, if \(DC\) is acyclic and \(x_1,\dots,x_n\) is an order compatible with
  \(DC\), then \(\bin[DC]\) is acyclic and \(x_1^1,\dots,x_1^b,\dots,x_n^1,\dots,x_n^b\) is an
  order compatible with \(\bin[DC]\). %
\end{lemma}
\begin{proof}
  The first part of the statement follows from the following observation:
  let \(\delta = (A,B,N)\) be a degree constraint and \(R\) a relation on
  variables \(e \supseteq B\) which respects \(\delta\), then \(\bin[R]\)
  respects \(\bin[\delta]\). %
  Indeed, let \(\tau\) be an assignment of \(\bin[A]\) and let \(\tau'\) be the
  corresponding assignment of \(A\) on domain \(2^b\) defined as \(\tau'(x)
  = \sum_{i=1}^b 2^{i-1}\tau(x^i)\). %
  Then it is easy to see that \(R[\tau']|_Y\) is in one-to-one correspondence
  with \(\bin[R][b][\tau]|_{\bin[Y]}\) by using the same encoding. %
  In particular, \(\size{\bin[R][b][\tau]|_{\bin[Y]}} = \size{R[\tau']|_Y} \leq
  N\). %

  Now let \(Q' \in \CDeg[\bin[DC]][\bin[H]]\) be a query on domain \(D\). %
  We let \(Q\) to be the query on hypergraph \(H\) on domain \(D^b\) where
  each relation \(R'\) of \(Q'\) on variables \(\bin[Y]\) is transformed
  into a relation \(R\) on variables \(Y\) as follows: for a tuple \(\tau'
  \in R'\), we build the tuple \(\tau \in R\) by taking for each \(y \in
  Y\), \(\tau(y) = \bigtimes_{j \leq b}\tau'(y^j)\). %
  It is easy to see that if \(R'\) respects degree constraint \(\bin[\delta]\),
  then \(R\) respects \(\delta\) and that \(\ans[Q]\) and \(\ans[Q']\) are
  in one-to-one correspondence. %
  Hence \(Q \in \CDeg\) and then \(\wc(\CDeg[\bin[DC]][\bin[H]]) \leq
  \wc(\CDeg)\). %

  Finally, by definition, it is clear that there is an edge in \(G_{DC}\)
  between \(x\) and \(y\) if and only if there is an edge between \(x^i\)
  and \(y^j\) for every \(i,j \leq b\) in \(G_{\bin[DC]}\). %
  Assume towards a contradiction that there is a path from \(x_i^j\) to
  \(x_k^\ell\) for some \(i \geq k\) in \(G_{\bin[DC]}\). %
  Then there is necessarely a path from \(x_i\) to \(x_k\) in \(G_{DC}\) from
  what precedes, which contradicts the fact that \((x_1,\dots,x_n)\) is a
  topological sort of \(G_{DC}\). %
\end{proof}

A direct consequence of \cref{lem:binariseDC} is that
\(\CDeg[\bin[DC]][\bin[H]]\) is \kl{prefix closed} and its worst case is
not greater than the worst case of \(\CDeg\). %
Since for any \(Q \in \CDeg\), the domain of \(\bin[Q]\) is two and has \(bn
= \bigot(n)\) variables, we have a worst case optimal join algorithm for
\(\CDeg\): %

\begin{corollary}
  \label{cor:wc}
  Let \(\CDeg\) be a class of \kl{join queries} defined for hypergraph
  \(H=(X,E)\), \(m = \size{E}, n = \size{X}\) and acyclic degree
  constraints \(DC\). %
  Assume \(x_1,\dots,x_n\) is an order compatible with \(DC\). %
  Then for every \(Q \in \CDeg\) on domain \(D \subseteq [2^b]\),
  \(\WCJ(\bin[Q],\emptytuple)\) on order
  \(x_1^1,\dots,x_1^b,\dots,x_n^1,\dots,x_n^b\) returns \(\ans\) in time
  \(\bigot(mn \cdot \wc(\CDeg))\). %
\end{corollary}

There is a slight abuse in the statement of \cref{cor:wc} as the algorithm
does not directly return \(\ans[Q]\) but a binary representation of each
tuple in \(\ans[Q]\). %
However, it is straightforward to turn each answer of \(\bin[Q]\) back to
the corresponding answer of \(Q\) in \(\bigot(1)\). %


\subparagraph*{Comparison with Generic Join and Leapfrog Triejoin.}

Our algorithm is quite similar to \emph{Generic Join}\cite[Algoritm~3]{NgoRR13} and
\emph{Leapfrog Trie Join}~\cite{veldhuizen2014triejoin}, which can already be seen
as a particular case of Generic Join. %
Similarly, \cref{alg:wcj} can be seen as a degenerated form of Generic
Join. %
The main difference in the approach is that both Generic Join and Trie Join
use a specific algorithm (trie join algorithm and \(m\)-way sort merge
respectively) to ensure that a variable \(x\) is branched only on values
that would not introduce any inconsistency. %
We circumvent this need by using binarisation instead, which can be seen,
from a higher perspective, as simply branching on the bits of each value
instead. %

The main novelty in our work is the complexity analysis. %
The analysis for Generic Join from~\cite{NgoRR13} relies on the knowledge of
the value of the worst case, known as the AGM bound for cardinality
constraints and polymatroid bound for acyclic degree constraints. %
While proofs of these bounds can be found in numerous references
(see~\cite{Suciu23} for a survey), they add a layer of complexity in the
understanding of why these simple branch and bound strategies achieve worst
case optimality. %
Our analysis instead relies on a very easy to check property of the classes
considered, namely the prefix closedness property: ``forgetting'' variables
in the tables will not allow to create an instance having more answers than
the worst case. %
Proving that classes defined by cardinality constraints or acyclic degree
constraints are prefix closed is elementary and very natural. %
Our approach is closer in spirit with the one taken in
\cite{veldhuizen2014triejoin} where the concept of \emph{renumbering} is
introduced, and the analysis of the runtime is bounded by the runtime of
the algorithm on a normalised instance where some values have been
changed. %
While the approach is similar, we feel that \kl{prefix closedness} is an
easier notion. %


\section{Uniform Sampling}
\label{sec:sampling}

There has been significant work around extending WCOJ algorithms into
sampling algorithms. %
It is known that if \(Q \in \C\) for \(\C\) a class of queries defined by
cardinality constraints~\cite{dengJSH2023,kimAGMOUT2023} or by acyclic degree
constraints~\cite{wangJSA2024}, one can uniformly sample \(\tau \in \ans\) in time
\(\bigot({\size{\wc(\C)} \over \max(1, \size{\ans})})\). %
We recover both results using an elementary algorithm, adaptated
from~\cite{rosenbaumSLTEP1993} to uniformly sample leaves from a tree without
exploring it completely. %

\subsection{Efficiently sampling the leaves of a tree}

Our core technique relies on the problem of sampling uniformly a leaf of a
rooted tree. %
We aim at designing an algorithm which avoids exploring the tree
exhaustively. %
This question has already been addressed by Rosenbaum
in~\cite{rosenbaumSLTEP1993} where he proposes an algorithm exploring the
tree in a top-down manner, and whenever it encouters a leaf, either returns
it or fails. %
To remove bias toward subtrees having many leaves in the algorithm, he
guides the search with upper bounds on the number of leaves of each
subtree, obtained from the depth and the branching size of the tree. %

We adapt Rosenbaum's algorithm in a slightly more general setting. %
Indeed, in our approach, the leaves of the tree will correspond to cases
where the recursion of \cref{alg:wcj} stops. %
In this case, either a solution is found and we are interested in the leaf,
or a inconsistency is found and we want to reject the leaf. %
We adapt the algorithm to be able to work on a tree when we want to sample
only a subset of its leaves. %
Moreover, the upper bounds on the leaves used in~\cite{rosenbaumSLTEP1993} is
too coarse for our purposes, we therefore describe the algorithm by using
oracle calls to a function (over-)estimating this number of leaves. %
This motivates the following definition: %

\begin{definition}
  Let \(T\) be a rooted tree. A \AP\intro{leaf estimator \(\up\)} for \(T\)
  is a function mapping nodes of \(T\) to positive values such that: %
  \begin{enumerate*}[label=(\roman*)]
  \item for every node \(t\) with children \(t_1,\dots,t_n\), \(\up(t) \geq \sum_{i = 1}^n
    \up(t_i)\); we call functions with this property
    \emph{tree-superadditive}; and
  \item if \(t\) is a leaf, then \(\up(t) \in \{0, 1\}\).
  \end{enumerate*}
\end{definition}

We denote by \(\up(T)\) the value of \(\up(r)\) where \(r\) is the root of
\(T\). %
Now, given a tree \(T\) and a leaf estimator \(\up\) for \(T\), we say that
a leaf \(\ell\) of \(T\) is a \AP\intro{\(1\)-leaf of \(T\)} if and only if
\(\up(\ell) = 1\) and denote by \(\onel\) the set of \kl{\(1\)-leaves} of
\(T\). %
Our goal is to uniformly sample \(\ell \in \onel\). %
Observe that since \(\up\) is a tree-superadditive function, for any node
\(t\), we have that \(\up(t)\) is an upper bound on the number of
\kl{\(1\)-leaves} below \(t\). %
We define \(\mathsf{children}(t)\) as the function that, given a node
\(t\), returns the list of its direct children. %

We define our algorithm recursively as follows: %

\begin{algorithm}[H]
  \caption{A variation of the Rosenbaum algorithm~\cite{rosenbaumSLTEP1993}}
  \label{alg:rosenbaum}
  Sample a leaf in the subtree root in $t$ as follows:
  \begin{itemize}
  \item if \(t\) is a leaf belonging to \(\onel\), output the leaf with
    probability \(1\);
  \item it \(t\) is any other leaf, fail with probability \(1\); and
  \item if \(t\) has children \(t_1,\dots,t_n\), recursively sample a
    \kl{\(1\)-leaf} in \(t_i\) with probability \(\frac{\up(t_i)}{\up(t)}\)
    and return it if the recursive call in \(t_i\) succeeds and fail
    otherwise. %
    Note that we may directly fail without recursively sampling with
    probability \(1 - \sum_i\frac{\up(t_i)}{\up(t)}\). %
  \end{itemize}
\end{algorithm}

\begin{theorem}
  \label{thm:sampling-once}
  Let \(T\) be a tree rooted in \(r\) and \(\up\) a \kl{leaf estimator} for
  \(T\). %
  Let \(\mathsf{out}\) be the output of \cref{alg:rosenbaum} on input
  \(r\). %
  Then, for any leaf \(\ell \in \onel\), we have that \cref{alg:rosenbaum} is a
  uniform Las Vegas sampler with guarantees: %
  \[
    \mathsf{Pr}(\mathsf{out} = \ell) = \frac{1}{\up(T)} \text{\qquad and \qquad}
  \mathsf{Pr}(\mathsf{out} = \mathsf{fail}) = 1 -
  \frac{\size{\onel}}{\up(T)}
  \]

  \cref{alg:rosenbaum} consists of \(\bigo(B \cdot \mathsf{depth}(T))\) calls
  to \(\up\), where \(B\) is the branching size of the tree, in
  \(\bigo(\mathsf{depth}(T))\) calls to the \children\ function. %
\end{theorem}
\begin{proof}
  We proceed by induction on the depth of the tree \(T\). %
  If the tree \(T\) is of depth \(1\), then it can be one of two cases:
  either it belongs to \(\onel\) and then \(\size{\onel} = 1\) and
  therefore it is trivial to see that the algorithm samples this leaf with
  probability \(\frac{1}{\up(T)} = 1\), or it does not belong to \(\onel\)
  and then there is nothing to sample, so the algorithm fails inevitably. %

  Supposing that the property holds for a tree \(T'\) of depth at most
  \(k\), if we now have a tree \(T\) of depth \(k + 1\), then it has
  children \((t_1,\dots,t_n)\) each of depth at most \(k\). %
  Then, by induction, \cref{alg:rosenbaum} samples from a given \(t_i\) with
  probability \(\frac{\up(t_i)}{\up(t)}\). If the recursive call has
  succeeded, then the algorithm has sampled a leaf from \(t_i\) with
  probability \(\frac{1}{\up(t_i)}\). Since the random choices are
  independent, the probability of outputing this leaf from \(t\) is
  \(\frac{\up(t_i)}{\up(t)}\times \frac{1}{\up(t_i)} =
  \frac{1}{\up(t)}\). %

  The complexity statement is straightforward. We need to evaluate the
  number of selected leaves in each subtree along a path from the root to a
  leaf, leading to \(\bigo(B \cdot \mathsf{depth}(T))\) calls to \(\up\)
  and for each node we visit, we need to find the list of children, leading
  to \(\bigo(\mathsf{depth}(T))\) calls to \(\mathsf{children}\).
\end{proof}

We can extend \cref{thm:sampling-once} with the following corollary: %

\begin{corollary}
  \label{cor:sampling_with_rosenbaum}
  Given a tree \(T\) with branching size \(B\) and oracle access to a leaf
  estimator function \(\up(\cdot)\), we can sample the leaves in \(\onel\)
  with uniform probability \(\frac{1}{\size{\onel}}\), when
  \(\size{\onel}>0\) or answer that \(\size{\onel}=0\). %
  This is done by repeating \cref{alg:rosenbaum} an expected
  \(\bigo(\frac{\up(T)}{\max(1,\size{\onel})})\) number of times and thus
  with an expected number of calls to \up\ in
  \(\bigo(\frac{\up(T)}{\max(1,\size{\onel})} \cdot B \cdot
  \mathsf{depth}(T))\) and to \children\ in
  \(\bigo(\frac{\up(T)}{\max(1,\size{\onel})} \cdot \mathsf{depth}(T))\). %
\end{corollary}

\begin{proof}
  We first treat the case where \(\size{\onel} = n > 0\). %
  \cref{alg:rosenbaum} can either fail or produce a leaf that has been
  sampled with uniform probability. It is a \emph{Las Vegas} algorithm. %
  It samples \onel\ of \(T\) with uniform probability
  \(\frac{1}{\up(T)}\). %
  When \(\size{\onel} = n > 0\), it thus outputs some data with probability
  \(\frac{n}{\up(T)}\) or fails with probability \(1 -
  \frac{n}{\up(T)}\). %
  Now if we repeat the algorithm until it succeeds, as each \onel\ has the
  same probability to be outputted in one repetition, they all have the
  same probability to be outputted at the end of this process. %
  In a nutshell, this procedure uniformly chooses amongst the \onel\ of
  \(T\) which all have probability \(\frac{1}{n}\) to be outputted. %
  Moreover, the number of repetitions of \cref{alg:rosenbaum} until it
  succeeds follows a geometric distribution and its expected value is thus
  \(\frac{\up(T)}{n}\). %

  The case where \(\size{\onel} = 0\) is a little trickier. %
  Since \cref{alg:rosenbaum} can only fail, repeating it would result in an
  infinite loop. %
  We can circumvent this in one of two ways. %

  Either we start a full exploration of \(T\) in parallel and if it does
  not find any \onel\ in \(T\), we stop running \cref{alg:rosenbaum} (if
  \cref{alg:rosenbaum} returns a 1-leaf, we stop the exploration). %
  A second method would be to improve \cref{alg:rosenbaum} by maintaining
  the parts of \(T\) that have already been explored and by updating the
  values of \(\up(t)\) for each subtree that is explored by using the
  information from its children. %
  Eventually, the algorithm explores \(T\) entirely. %
  Indeed parts of \(T\) that have been explored are known not to contain
  data and have thus probability 0 to be explored again. %
  In the end, updating the \(\up(t)\) value of each node will result in
  having \(\up(T) = 0\), meaning that \onel\ is empty. %
  We then stop the search. %
  Both methods would cost a time of \(\bigo(\up(T))\). %
\end{proof}

\begin{remark}
  Notice that the second method presented in the proof may also be useful
  when \(\size{\onel}\) is small compared to \(\up(T)\). %
  Indeed, updating \(\up(\cdot)\) at each failure of the algorithm increases
  the probability of success of the next iteration of the algorithm
  resulting in an overall improvement in the speed of convergence. %
\end{remark}

\subsection{Applying \cref{alg:rosenbaum} to join queries}

There is a strong link between the tree structure used in the
aforementioned sampling method and \cref{alg:wcj}. %
Assume that we want to sample uniformly the results of the query \(Q\) with
relations over variables \(X =\{x_1,\dots, x_n\}\) and domain \(D\). %
For this, we can follow the structure of the execution of \cref{alg:wcj}
when it uses the order \((x_1,\dots,x_n)\) on variables. %
The execution of \cref{alg:wcj} naturally constructs a tree structure whose
nodes are some assignments of \(D^{X_i}\) for some \(i\in[0,n]\)
corresponding to the input of recursive calls. %
We call this tree the \emph{trace tree of \(Q\)} and denote it by
\(T_Q\). %
In \(T_Q\), an assignment \(\tau \in D^{X_i}\) is the \emph{parent} of another
assignment \(\tau'\) when \(\tau'= \tau \cup \unittuple[x_{i+1}][d]\). %
Among all the possible assignments, the ones that are nodes in \(T_Q\) are
those that are consistent with \(Q\) or those that are inconsistent with
\(Q\) but have a parent that is consistent with \(Q\). %
\cref{fig:example_dpll} depicts such a tree: assignments that are
inconsistent with \(Q\) are labelled \(\bot\), those that are elements of
the answer set are labelled \(\top\). %
Moreover, the assignment corresponding to a node in \cref{fig:example_dpll}
can simply be read off the path from the root to that particular node. %

The leaves of \(T_Q\) that we want to sample in this tree are simply the
elements of \(D^{X_n}\) that are consistent with \(Q\), namely the
solutions of \(Q\). %
\cref{alg:rosenbaum} can then be applied to \(T_Q\). %
Thus we call \AP\intro{Q-estimator} a function \qup\ on the nodes of \(T_Q\)
that verifies: %
\begin{itemize}
\item when \(\tau\) is a node of \(T_Q\) in \(D^{X_i}\) that is consistent with
  \(Q\), \(\qup(\tau) \geq \sum_{d\in D} \qup(\tau \cup \unittuple[x_{i+1}][d]
  )\) (i.e. \qup\ is tree-superadditive), %
\item \(\qup(\tau) = 1\) when \(\tau \in \ans\)
\item \(\qup(\tau) = 0\) when \(\tau\) is inconsistent with \(Q\).
\end{itemize}

\begin{theorem}
  \label{th:tree-superadd-query}
  Given a \kl{\(Q\)-estimator} \(\qup(\tau)\) that can be evaluated in time
  \(t\) for every \(\tau\), it is possible to uniformly sample \ans\ with
  expected time \(\bigo(\frac{\qup(\emptytuple)}{\max(1,\size{\ans})} \cdot
  \size{X} \cdot \size{D} \cdot t)\). %
\end{theorem}
\begin{proof}
  This is a consequence of \cref{cor:sampling_with_rosenbaum}. %
  The depth of \(T_Q\) is \(\size{X}\) and its branching size is
  \size{D}. %
\end{proof}

\subsection{Tree-superadditive worst-case bounds}
\label{ssec:wc-bounds}

It now remains to define \kl{\(Q\)-estimators} when \(Q\) belongs to
\CCard\ or to \CDeg. %
This will allow us to recover the sampling results from the
literature~\cite{dengJSH2023,kimAGMOUT2023,wangJSA2024} by using
\cref{th:tree-superadd-query} and the binarisation technique
\cref{sec:binarisation}. %
As the class \CDeg\ generalises the class \CCard, we could have only treated
the first case. %
We think however that the \CCard\ being simpler, it conveys more
intuition. %
Our main tool is a simple consequence of an inequality by
Friedgut~\cite{friedgut04}. %

\begin{lemma}[Friedgut, {\cite[Lemma 3.3]{friedgut04}}]
  \label{lem:friedgut_consequence}
  For every finite sets \(I\) and \(J\), every family of positive real numbers
  \((\omega_j)_{j\in J}\) so that \(\sum_{j\in J}\omega_j \geq 1\), and every family of
  positive real numbers \((a_{i,j})_{i\in I, j\in J}\), we have:
  \[\sum_{i\in I} \prod_{j\in J} a_{i,j}^{\omega_j} \leq \prod_{j\in J}\left( \sum_{i\in I} a_{i,j} \right)^{\omega_j}\enspace .\]
\end{lemma}

\begin{proof}
  This is a consequence of the \emph{Generalised Weighted Entropy Lemma}
  in~\cite[Lemma 3.3]{friedgut04}. The statement of this lemma is depends
  on the following objects: %
  \begin{itemize}
  \item A hypergraph \(H=(X,E)\).
  \item A finite set \(L\).
  \item A family of subsets of \(X\) \((F_l)_{l\in L}\). %
    Let \(e_l = e\cap F_l\) for every \(e\in E\). And let \(E_l = \{e_l\mid e\in
    E\}\). %
  \item A family of weights \(W = (w_l)_{l\in L}\): \(w_l\) associates a positive
    real number to the elements of \(E_l\). %
  \item A family of positive real numbers \(A = (\alpha_{l})_{l\in L}\) so that for
    every \(x\in X\), \(\sum_{l \mid x\in F_l} \alpha_l\geq 1\). %
  \end{itemize}
  Then it states that:
  \[
    \sum_{e\in E} \prod_{l\in L} w_l(e_l) \leq \prod_{l\in L}\left( \sum_{e_l\in E_l} w_l(e_l)^{1/\alpha_l} \right)^{\alpha_l}
    \enspace.
  \]

  The statement of the lemma we want to prove is obtained by setting:
  \begin{itemize}
  \item \(X = I\), \(E = \{\{i\} \mid i \in I\}\),
  \item \(L = J\) and for every \(j\in J\), \(F_j = I\), as a consequence for
    every \(e\in E\) and \(j\in J\), \(e_j = e\), and thus \(E_j = E\). %
  \item For each \(j\in J\), we let \(w_j(\{i\}) = a_{i,j}^{\omega_j}\). %
  \item Finally, we let \(A = (\omega_j)_{j\in J}\). %
    As for every \(i\in I\) and \(j\in J\), we have that \(i\in F_j\), the
    hypothesis that \(A\) must satisfy is a consequence of the hypothesis
    \(\sum_{j\in J}\omega_j \geq 1\). %
  \end{itemize}
  In this setting, the \emph{Generalised Weighted Entropy Lemma} gives us:
  \[
    \sum_{i\in I} \prod_{j\in J} a_{i,j}^{\omega_j} \leq %
    \prod_{j\in J}\left( \sum_{i\in I} (a_{i,j}^{\omega_j})^{1/\omega_j}
    \right)^{\omega_j} = %
    \prod_{j\in J}\left( \sum_{i\in I} a_{i,j} \right)^{\omega_j} \enspace .
  \]

  Which is the expected inequality.
\end{proof}

\subparagraph*{Cardinality constraints.}

Let \(\CCard\) be a class of join queries defined for hypergraph
\(H=(X,E)\) (with \(X = \{x_1,\dots,x_n\}\) and \(E = \{e_1,\dots,e_m\}\))
and cardinality constraints \(\tup[N] \subseteq \N^E\). %
In \cite{atseriasSB2013}, Atserias, Grohe and Marx show that \(\wc(\CCard)\)
can be computed from the solutions of the following linear program: %
\begin{align*}
  \label{eq:fec}
  \mathsf{min} \sum_{j}^m \omega_j\log(N(e_j))\  s. t., \forall i = 1,\dots,n & \sum_{j : x_i \in E_j} \omega_j \geq 1
\end{align*}
and prove \(\wc(\CCard)\) is \(\prod_{j=1}^m N(e_j)^{\omega_j}\) up to polylogarithmic factors. %
The vector \(\omega\) is called a \emph{fractional cover} of \(H\). %

Let us fix a fractional cover \(\omega\) of \(H\), for a query \(Q\) in
\(\CCard\). %
For every \(j \in [1,m]\), we let \(R_j \in Q\) be a relation such that
\(\size{R_j}\leq N(e_j)\) (which exists by definition of \(\CCard\)). %
We take as \kl{\(Q\)-estimator} \(\agmup(\tau)\) as follows: %
\[
  \agmup(\tau) =
  \left\{
    \begin{array}{ll}
      0%
      & \text{when \(\tau\) is inconsistent with } Q\\%
      \prod_{j=1}^{m} \size{R_j[\tau]}^{\omega_j}%
      & \text{otherwise} %
    \end{array}
  \right.
  \enspace.
\]

When \(\tau\) is inconsistent with \(Q\), then, by definition, \(\agmup(\tau) =
0\). %
Furthermore, when \(\tau\) is in \ans, for every \(j\), \(\size{R_j[\tau]}=1\) and
therefore, \(\agmup(\tau) =1\). %
To show that \(\agmup\) is a \kl{\(Q\)-estimator}, we finally need to show
that for every \(\tau\) that is consistent with \(Q\), we have
\(\agmup(\tau) \geq \sum_{d\in D} \agmup(\tau \cup \unittuple[x_{i+1}][d]
)\). %

For every \(j\) we have:
\begin{itemize}
\item \(\size{R_j[\tau]} = \sum_{d\in D} \Size{R_j[\tau \cup
    \unittuple[x_{i+1}][d]]}\) when \(x_{i+1}\) is in \(X_{R_j} - X_i\)
\item \(\size{R_j[\tau]} = \size{R_j[\tau \cup \unittuple[x_{i+1}][d]]}\) for
  every \(d \in D\) otherwise.
\end{itemize}
We let \(K = \{k \in [1,m] \mid x_{i+1} \in X_{R_j} - X_i\}\) and \(L = [1,m]\setminus K\).
Since \(\bigcup E = X\), we must have \(K \neq \emptyset\). %
We thus have: %
\begin{align*}
  \sum_{d\in D} \agmup(\tau \cup \unittuple[x_{i+1}][d])
  &\,= \sum_{d\in D} \prod_{k\in K} \size{R_k[\tau \cup \unittuple[x_{i+1}][d]]}^{\omega_k} \times \prod_{l\in L}\size{R_l[\tau \cup \unittuple[x_{i+1}][d]]}^{\omega_l}\\
  &\,= \sum_{d\in D} \prod_{k\in K} \size{R_k[\tau \cup \unittuple[x_{i+1}][d]]}^{\omega_k} \times \prod_{l\in L}\size{R_l[\tau]}^{\omega_l}\\
  &\,=  \prod_{l\in L}\size{R_l[\tau]}^{\omega_l} \times \sum_{d\in D} \prod_{k\in K} \size{R_k[\tau \cup \unittuple[x_{i+1}][d]]}^{\omega_k}
\end{align*}
Then \(\agmup(\tau) \geq \sum_{d\in D} \agmup(\tau \cup \unittuple[x_{i+1}][d] )\) follows
from: %

\[
  \sum_{d\in D}\prod_{k\in K}\Size{R_k[\tau \cup \unittuple[x_{i+1}][d]]}^{\omega_k} \leq \prod_{k\in
    K}\left(\sum_{d\in D} \Size{R_k[\tau \cup
      \unittuple[x_{i+1}][d]}\right)^{\omega_k} \enspace .
\]

By definition of \(K\) and since \((\omega_j)\) is a fractional cover of \(H\),
we have \(\sum_{k \in K} \omega_k \geq 1\). %
Hence we can directly get the bound using
\cref{lem:friedgut_consequence}. %

\begin{theorem}
 \label{th:agm_sampling}
 Given \(\CCard\) a class of join queries defined for hypergraph
 \(H=(X,E)\) and cardinality constraints \(\tup[N] \subseteq \N^E\), for
 every query \(Q\) in \CCard, it is possible to uniformly sample \ans\ with
 expected time \(\bigot(\frac{\wc(\CCard)}{\max(1,\size{\ans})} \cdot
 \size{X} \cdot \log(\size{D}) \cdot \size{E})\). %
\end{theorem}
\begin{proof}
  Given \(Q\) with active domain \(D\), we let \(b = \lceil \log(\size{D})\rceil\)
  and we are going to sample \(\bin[Q]\). %
  As we have seen in \cref{sec:binarisation}, \ans\ and \(\ans[\bin[Q]]\)
  can be considered to be the same set. %
  Moreover \(\bin[Q]\) belongs to a class \(\bin[\C]\) defined by
  cardinality constraints where \(\wc(\bin[\C]) \leq \wc(\CCard)\). %
  Therefore, there is a \(\bin[Q]\)-estimator \(\agmup(\cdot)\). %
  Using the data structure described in \cref{sec:dpll} to represent
  \(R[\tau]\) or a trie structure annotated with cardinalities to represent
  every relation of \(Q\), we can compute \(\Size{R_k[\tau \cup
    \unittuple[x_{i+1}]]}\) in constant time. %
  Therefore, computing \(\agmup(\cdot)\) takes time \(\bigot(|E|)\). %
  Finally, by definition of \(\agmup(\cdot)\), we have \(\agmup(\emptytuple) \leq
  \wc(\bin[\C])\). %
  This allows us to apply \cref{th:tree-superadd-query} and yields the
  claimed complexity. %
\end{proof}

\subparagraph*{Polymatroid.}

Let \(\CDeg\) be a class of join queries defined for hypergraph \(H=(X,E)\)
(with \(X = \{x_1,\dots,x_n\}\) and \(E = \{e_1,\dots,e_m\}\)) and acyclic
degree constraints \((A_\delta, B_\delta, N_\delta)_{\delta \in DC}\). %
The polymatroid bound is a generalisation of the AGM bound that can be
formulated on acyclic degree constraints with the solutions of the
following linear program: %

\begin{equation}
  \label{eq:poly}
  \mathsf{min} \sum_{\delta \in DC} \omega_\delta \log(N_\delta) \quad
  s. t., \forall x \in X,\, \sum_{\delta : x \in B_\delta \setminus A_\delta} \omega_\delta \geq 1\enspace .
\end{equation}

For this program to have a solution, we need to assume that for every \(x \in
X\), there is at least one constraint \(\delta\) such that \(x \in B_\delta
\setminus A_\delta\). %
As stated in \cref{sec:preliminaries}, to ensure that acyclic degree
constraints induce a finite worst case, we generally assume that \(\bigcup
E = X\) and that for each \(e \in E\), we have at least one cardinality
constraint \((\emptyset, e, N_e)\). %
We assume this condition to be met here. %
Now for any solution \(\omega_\delta\) of the previous program, it has been
shown~\cite{ngoWCOJ2018} that \(\wc(\CDeg)\) is \(\prod_{\delta \in DC}
N_\delta^{\omega_\delta}\) up to some polylogarithmic factors. %

Let us fix a solution \(\omega\) of \eqref{eq:poly}, and an order
\((x_1,\dots,x_n)\) on \(X\) compatible with \(DC\). %
For a query \(Q\) in \CDeg\ and \(\delta=(A_\delta, B_\delta, N_\delta)\) in \(DC\), we let
\(R_\delta\) to be a relation in \(Q\) that respects \(\delta\) and such
that \(X_{R_\delta} = e_\delta\) (in other words, the relation guards this
constraint). %
To lighten notations, we denote by \(R'_\delta = R_\delta|_{B_\delta}\), since the degree
constraint is applied to this projection of \(R_\delta\) and not
\(R_\delta\) itself. %
Moreover, given \(\tau\) in \(D^{X_i}\), we let:
\[
  N_\delta[\tau] = %
  \left\{
    \begin{array}{ll}
      0&\text{if \(\tau\) is inconsistent with } Q\\
      N_\delta & \text{if } A_\delta \setminus X_i \neq \emptyset\\
      \size{R'_\delta[\tau]}& \text{otherwise}
    \end{array}
  \right.
\]

We now take as \kl{\(Q\)-estimator} \(\pmup(\tau)\) as follows:
\[
  \pmup(\tau) = \prod_{\delta\in DC} N_\delta[\tau]^{\omega_\delta} \enspace.
\]

When \(\tau\) is inconsistent with \(Q\), by definition of \(N_\delta[\tau]\),
\(\pmup(\tau) = 0\). %
If \(\tau\) is in \ans, since for every \(\delta\), \(A_\delta\setminus X_n = \emptyset\),
\(N_\delta[\tau] = \size{R'_j[\tau]}=1\) and therefore, \(\pmup(\tau)
=1\). %
Proving that \(\pmup\) is a \(Q\)-estimator finally requires to show that
for every \(\tau\) that is consistent with \(Q\), we have \(\pmup(\tau)
\geq \sum_{d\in D} \pmup(\tau \cup \unittuple[x_{i+1}][d] )\). %

We let \(K = \{\delta \in DC \mid x_{i+1} \in B_\delta \setminus A_\delta\}\) and \(L = DC\setminus
K\). %
As stated before, for the linear program to have a solution, we assumed \(K
\neq \emptyset\). %
We make two observations:
\begin{enumerate}
\item For \(\delta \in K\), \(N_\delta[\tau] = \size{R'_\delta[\tau]}\) and for any \(d \in D\), \(N_\delta[\tau \cup
  \unittuple[x_{i+1}]] = \size{R'_\delta[\tau \cup \unittuple[x_{i+1}]]}\),
\item For \(d \in D\), \(N_\delta[\tau \cup \unittuple[x_{i+1}]] \leq N_\delta[\tau]\).
\end{enumerate}

The first observation follows directly from the facts that \(x_{i+1} \in
B_\delta \setminus A_\delta\) by definition of \(K\) and that
\(x_1,\dots,x_n\) is compatible with \(DC\) which implies that for any
\(\delta \in K\), \(A_\delta \subseteq X_i\). %
For the second inequality, first assume \(A_\delta \setminus X_{i+1} \neq
\emptyset\). %
Then both sides are equal to \(N_\delta\). %
Now assume \(A_\delta \setminus X_{i+1}=\emptyset\). %
Then \(N_{\delta}[\tau \cup \unittuple[x_{i+1}]] = \size{R'_\delta[\tau \cup \unittuple[x_{i+1}]]}\)
and either \(N_\delta[\tau] = \size{R'_\delta[\tau]}\), then the equality
is clear, or \(N_\delta[\tau]=N_\delta\) and the inequality follows from
the fact that \(R_\delta\) respects \(\delta\) by definition. %
Then we obtain:
\begin{align*}
  \sum_{d\in D} \pmup(\tau \cup \unittuple[x_{i+1}][d])
  &\,= \sum_{d\in D} \prod_{\delta\in K} \size{R'_\delta[\tau \cup \unittuple[x_{i+1}][d]]}^{\omega_\delta} \times \prod_{\delta\in L}N_\delta[\tau\cup\unittuple[x_{i+1}][d]]^{\omega_\delta}\\
  &\,\leq  \sum_{d\in D} \prod_{\delta\in K} \size{R'_\delta[\tau \cup \unittuple[x_{i+1}][d]]}^{\omega_\delta} \times \prod_{\delta\in L}N_\delta[\tau]^{\omega_\delta}\\
  &\,= \prod_{\delta\in L}N_\delta[\tau]^{\omega_\delta} \times \sum_{d\in D} \prod_{\delta\in K} \size{R'_\delta[\tau \cup \unittuple[x_{i+1}][d]]}^{\omega_\delta}\\
\end{align*}

We hence get \(\sum_{d\in D} \pmup(\tau \cup \unittuple[x_{i+1}][d]) \leq \pmup(\tau)\) by
showing that: %

\[
  \sum_{d\in D}\prod_{\delta\in K}\Size{R'_\delta[\tau\cup \unittuple[x_{i+1}][d]]}^{\omega_\delta} \leq \prod_{\delta\in
    K}\left(\sum_{d\in D} \Size{R'_\delta[\tau \cup
      \unittuple[x_{i+1}][d]}\right)^{\omega_\delta} = \prod_{\delta\in K}
  R'_{\delta}[\tau]^{\omega_\delta} = \prod_{\delta\in K}
  N_{\delta}[\tau]^{\omega_\delta} \enspace .
\]

The above inequality is a consequence of \cref{lem:friedgut_consequence}. %
Indeed, \(\sum_{\delta \in K}\omega_\delta \geq 1\) by definition of \(K\) and because of the
constraints \(\omega_\delta\) verifies in the linear program. %
We then conclude with this result, obtained exactly in the same way as
\cref{th:agm_sampling}. %

\begin{theorem}
 \label{th:pm_sampling}
 Given \(\CDeg\) a class of join queries defined for hypergraph \(H=(X,E)\)
 and acyclic degree constraints \(DC\), for every query \(Q\) in \CDeg, it
 is possible to uniformly sample \ans\ with expected time
 \(\bigo(\frac{\wc(\CDeg)}{\max(1,\size{\ans})} \cdot \size{X} \cdot
 \log(\size{D}) \cdot \size{E})\).
\end{theorem}



\newpage

\bibliography{biblio}


\end{document}